\documentclass[12pt]{amsart}

\usepackage{tgpagella}
\usepackage{euler}
\usepackage[T1]{fontenc}
\usepackage{amsmath, amssymb}
\usepackage[hidelinks]{hyperref}
\usepackage[english]{babel}
\usepackage{mathrsfs}
\usepackage{eucal}
\usepackage[all]{xy}
\usepackage{tikz}

\newtheorem{thm}{Theorem}[section]
\newtheorem*{thm*}{Theorem}
\newtheorem{lem}[thm]{Lemma}
\newtheorem{fact}[thm]{Fact}

\newtheorem{prop}[thm]{Proposition}
\newtheorem*{prop*}{Proposition}

\newtheorem*{cor*}{Corollary}

\theoremstyle{definition}
\newtheorem{defn}[thm]{Definition}
\newtheorem*{defn*}{Definition}

\newtheorem{question}[thm]{Question}

\newtheorem*{question*}{Question}
\newtheorem*{Pquestion*}{Popa's question}

\newtheorem*{conv*}{Convention}

\def\bb{\mathbb}

\def\bb{\mathbb}

\newcommand{\cstar}{$\mathrm{C}^*$}

\makeatletter

\def\dotminussym#1#2{%
  \setbox0=\hbox{$\m@th#1-$}%
  \kern.5\wd0%
  \hbox to 0pt{\hss\hbox{$\m@th#1-$}\hss}%
  \raise.6\ht0\hbox to 0pt{\hss$\m@th#1.$\hss}%
  \kern.5\wd0}

\def \val{\operatorname{val}}
\def \pval{\operatorname{s-val}}
\newcommand{\mip}{\operatorname{MIP}}

\newcommand{\cqs}{C_q}
\newcommand{\cqa}{C_{qa}}
\newcommand{\cqc}{C_{qc}}

\begin{document}

%%%%%%%%%%%%%%%%%%%%%%%%%%%%%%%%%%%%%%%%%%%%%%

\title{Approximate traces on groups and the quantum complexity class $\mip^{co,s}$}
\author{Isaac Goldbring and Bradd Hart}

\thanks{I. Goldbring was partially supported by NSF grant DMS-2504477.}

\address{Department of Mathematics\\University of California, Irvine, 340 Rowland Hall (Bldg.\# 400),
Irvine, CA 92697-3875}
\email{isaac@math.uci.edu}
\urladdr{http://www.math.uci.edu/~isaac}

\address{Department of Mathematics and Statistics, McMaster University, 1280 Main St., Hamilton ON, Canada L8S 4K1}
\email{hartb@mcmaster.ca}
\urladdr{http://ms.mcmaster.ca/~bradd/}

\begin{abstract}
    An open question in quantum complexity theory is whether or not the class $\mip^{co}$, consisting of languages that can be efficiently verified using interacting provers sharing quantum resources according to the quantum commuting model, coincides with the class $coRE$ of languages with recursively enumerable complement.  We introduce the notion of a qc-modulus, which encodes approximations to quantum commuting correlations, and show that the existence of a computable qc-modulus gives a negative answer to a natural variant of the aforementioned question.
\end{abstract}

\maketitle

\section{Introduction}

Throughout this note, $n$ and $m$ denote natural numbers that are at least $2$ and $[n]$ denotes the set $\{1,\ldots,n\}$ (and likewise for $[m]$).

We recall the following definitions from quantum information theory and quantum complexity theory.

\begin{defn}

\

\begin{enumerate}
\item The set $\cqs(n,m)$ of \textbf{quantum correlations} consists of the correlations of the form $p(i,j|v,w)= \langle A^v_i \otimes B^w_j\xi,\xi \rangle$ for $v,w\in [n]$ and $i,j\in [m]$,
where H is a finite-dimensional Hilbert space, $\xi \in H \otimes H$ is a unit vector, and for every $v,w \in [n]$,
$(A^v_i: i\in [m])$ and $(B^w_j \ : \ j\in [m])$ are positive operator-valued measures (POVMs) on H.
\item We set $\cqa(n,m)$ to be the closure in $[0,1]^{n^2m^2}$ of $\cqs(n,m)$.
\item The set $\cqc(n,m)$ of \textbf{quantum commuting correlations} consists of the correlations of the form $p(i,j|v,w)= \langle A^v_i B^w_j\xi,\xi \rangle$ for $v,w \in [n]$ and $i,j\in [m]$,
where H is a separable Hilbert space, $\xi \in H$ is a unit vector, and for every $v,w\in [n]$,
$(A^v_i: i\in [m])$ and $(B^w_j \ : \ j\in [m])$ are POVMs on H for which $A^v_i B^w_j=B^w_jA^v_i$ for all $v,w\in [n]$ and $i,j\in [m]$.
\end{enumerate}
\end{defn}

Note that $C_{q}(n,m)\subseteq C_{qc}(n,m)$; since it is know that the latter set is closed, we in fact have that $C_{qa}(n,m)\subseteq C_{qc}(n,m)$.  Tsirelson's problem in quantum information theory asks if these latter two sets in fact coincide; a negative resolution to this problem follows from the main result of \cite{MIP*}.

\begin{defn}

\

\begin{enumerate}
    \item A \textbf{nonlocal game with $n$ questions and $m$ answers} is a pair $\frak G=(\mu,D)$, where $\mu$ is a probability distribution on $[n]\times [n]$ and $$D:[n]\times [n]\times [m]\times [m]\to \{0,1\}$$ is a function.
\item Given a nonlocal game $\frak G$ as in the previous item and $p\in [0,1]^{n^2m^2}$, we define the \textbf{value of $\frak G$ at $p$} to be $$\val(\frak G,p):=\sum_{v,w\in [n]}\mu(v,w)\sum_{i,j\in [m]}D(v,w,i,j)p(i,j|v,w).$$
\end{enumerate}
\end{defn}

\begin{defn}
Given a nonlocal game $\frak G$ with $n$ questions and $m$ answers, we define:
\begin{enumerate}
    \item the \textbf{entangled value of $\frak G$} to be 
    $$\val^*(\frak G):=\sup_{p\in \cqa(n,m)}\val(\frak G,p);$$
    \item the \textbf{quantum commuting value of $\frak G$} to be $$\val^{co}(\frak G):=\sup_{p\in \cqc(n,m)}\val(\frak G,p).$$
\end{enumerate}
\end{defn}

\begin{defn}
A language $L$ (in the sense of complexity theory) belongs to the class $\mip^*$ if there is an effective mapping $z\mapsto \frak G_z$ from strings to nonlocal games such that:
\begin{itemize}
    \item if $z\in L$, then $\val^*(\frak G_z)\geq \frac{2}{3}$
    \item if $z\notin L$, then $\val^*(\frak G_z)\leq \frac{1}{3}$.
\end{itemize}
The complexity class $\mip^{co}$ is defined in the analogous fashion, using $\val^{co}$ instead of $\val^*$.
\end{defn}

The following landmark result in quantum complexity theory appears in \cite{MIP*}:

\begin{thm}
$\mip^*=\operatorname{RE}$.  In other words, the languages that belong to $\mip^*$ are precisely the recursively enumerable languages, that is, those languages $L$ for which there is an algorithm that enumerates $L$.
\end{thm}

In particular, $\mip^*$ contains undecidable problems (such as the halting problem).

In \cite{MIP*}, it was pointed out that every element of the class $\mip^{co}$ is co-recursively enumerable, that is, the complement of $\mip^{co}$ is recursively enumerable.  Denoting by $\operatorname{coRE}$ the complexity class of langauges that are co-recursively enumerable, the authors of \cite{MIP*} ask whether or not the aforementioned inclusion is actually an equality:

\begin{question}\label{mainquestion}
Does $\mip^{co}=\operatorname{coRE}$?
\end{question}

The main result of this note is that a very natural variant of the above question has a negative answer provided a certain computability assumption about ``approximate'' quantum commuting correlations holds, and that, in fact, all languages that belong to this natural variant of $\mip^{co}$ are actually decidable.

First, we recall the following definitions:

\begin{defn}

\

\begin{enumerate}
\item Given $p\in [0,1]^{n^2m^2}$, we say that $p$ is \textbf{synchronous} if $p(i,j|v,v)=0$ for all $v\in [n]$ and all \emph{distinct} $i,j\in [m]$.  
\item We let $\cqa^s(n,m)$ and $\cqc^s(n,m)$ denote the synchronous elements of $\cqa(n,m)$ and $\cqc(n,m)$ respectively.  \item Given a nonlocal game $\frak G$, we let $\pval^*(\frak G)$ and $\pval^{co}(\frak G)$ denote the corresponding \textbf{synchronous values} of $\frak G$, which are defined analogously to $\val^*(\frak G)$ and $\val^{co}(\frak G)$, except that we only take the supremum over $C^s_{qa}(n,m)$ and $C^s_{qc}(n,m)$ respectively.
\item We define the complexity class $\mip^{*,s}$ and $\mip^{co,s}$ analogously to $\mip^*$ and $\mip^{co}$, this time using the appropriate synchronous values of the games in the definition.
\end{enumerate}
\end{defn}

The main result of \cite{MIP*} actually shows that $\mip^{*,s}$ coincides with $RE$ and thus a reasonable variant of Question \ref{mainquestion} above is whether or not $\mip^{co,s}$ coincides with coRE.

We will see later in Proposition \ref{lemma2} that synchronous commuting correlations satisfy a certain ``stability'' property, namely that correlations that are almost quantum commuting correlations (in a certain technical sense) are near actual quantum commuting correlations.  The main result of this note (Theorem \ref{maintheorem} below) will show that if this stability relation can be realized ``effectively,'' then all languages in $\mip^{co,s}$ are in RE, and are thus decidable.  Since there are elements of coRE that are undecidable (such as the complement of the halting problem), we would obtain a negative solution to the synchronous version of Question \ref{mainquestion} above.  Said in the opposite direction:  if it turns out that $\mip^{co,s}=$ coRE, then there is no effective version of stability for almost quantum commuting correlations. 

The authors would like to thank Alec Fox, Thomas Vidick, and Henry Yuen for helpful conversations around this work.

\section{Preliminaries: traces on groups and group \cstar-algebras}

Fix a countable group $G$ and let $\bb CG$ denote the corresponding group ring.  We recall the following terminology:

\begin{defn}
A function $\tau:G\to \bb C$ is called:
\begin{enumerate}
    \item \textbf{of positive type} if for all $\sum_{\lambda}a_\lambda u_\lambda\in \bb CG$, we have $\sum_{\lambda,\gamma}\overline{a_\lambda}a_\gamma \tau(\lambda^{-1}\gamma)\geq 0$
    \item a \textbf{class function} if $\tau$ is constant on conjugacy classes.
\end{enumerate}
\end{defn}

The following terminology is not standard, but convenient for our purposes.  In what follows, $\bb D$ denotes the unit disc in the complex plane.

\begin{defn}\label{tracedef}
A function $\tau:G\to \bb D$ is called a \textbf{trace} on $G$ if it is a class function of positive type.
\end{defn}

Here is the reason for the abuse in terminology.  Below, $C^*(G)$ denotes the universal \cstar-algebra of the group $G$.

\begin{fact}
Given a trace $\tau$ on $C^*(G)$, its restriction to $G$ is a trace on $G$ (in the sense of Definition \ref{tracedef}).  Moreover, the map $\tau\mapsto \tau|G$ is a bijection between traces on $C^*(G)$ and traces on $G$.
\end{fact}

In the sequel, we will freely abuse notation and use $\tau$ to denote both the trace on the group $G$ as well as the corresponding trace on $C^*(G)$.

Note that in the definition of being positive type, we can restrict attention to elements of the subring $\bb Q(i)G$ without changing the notion.  Assuming then that some countable enumeration $G=(g_0,g_1,g_2,\ldots)$ of $G$ has been given, there is thus an effectively enumerable countable list $(R_l)$ of requirements that characterize when a function $\tau:G\to \bb D$ is a trace on $G$.  Note really that this list of requirements is independent of the group $G$ in question and just depends on some fixed effective coding of $\bb Q(i)$.  

We now consider ``relaxations'' of these requirements:

\begin{defn}
Fix $k,l\in \bb N$ with $k\geq 1$.
\begin{enumerate}
    \item If $R_l$ is the requirement $\sum_{\lambda,\gamma}\overline{a_\lambda}a_\gamma \tau(\lambda^{-1}\gamma)\geq 0$, then we define the relaxed requirement $R_l^k$ to be that $\sum_{\lambda,\gamma}\overline{a_\lambda}a_\gamma \tau(\lambda^{-1}\gamma)$ is within $\frac{1}{k}$ of the positive real axis.
    \item If $R_l$ is the requirement $\tau(\gamma^{-1}\lambda\gamma)=\tau(\lambda)$, then we define the relaxed requirement $R_l^k$ to be $|\tau(\gamma^{-1}\lambda\gamma)-\tau(\lambda)|<\frac{1}{k}$.
\end{enumerate}
We say that $\tau:G\to \bb D$ is a \textbf{$k$-approximate trace on $G$} if the relaxed requirements $R_1^k,\ldots,R_k^k$ hold.
\end{defn}

The following lemma is obvious:

\begin{lem}\label{lemma1}
For each $k\in \bb N$, there is $\delta>0$ such that, for all functions $\tau,\tau':G\to \bb D$, if $\tau$ is a trace on $G$ and $\|\tau-\tau'\|_\infty<\delta$, then $\tau'$ is a $k$-approximate trace on $G$.  Moreover, $\delta$ depends only on $k$ and not on $G$ and this dependence is computable from $k$.
\end{lem}

\section{The groups $\bb F(n,m)$}

Below, for $n,m\geq 2$, we let $\bb F(n,m)$ denote the group freely generated by $n$ elements of order $m$.  In the \cstar-algebra $C^*(\bb F(n,m))$, for each $v\in [n]$, we let $e^{n,m}_{v,1},\ldots,e^{n,m}_{v,m}$ denote the projections onto the eigenspaces corresponding to the eigenvalue $\xi_m^i$ of the unitary operator corresponding to the $v^{\text{th}}$ generator $u^{n,m}_v$ of $\bb F(n,m)$, where $\xi_m$ denotes a primitive $m^{\text{th}}$ root of unity. We then have that $(u_v^{n,m})^j=\sum_{i=1}^m \xi_m^{ji} e^{n,m}_{v,i}$ for each $v\in [n]$ and $j\in [m]$.

\begin{defn}
Given $p\in [0,1]^{n^2m^2}$, we say that a trace $\tau$ on $\bb F(n,m)$ is \textbf{adapted to $p$} if $p(i,j|v,w)=\tau(e^{n,m}_{v,i}e^{n,m}_{w,j})$ for all $v,w\in [n]$ and all $i,j\in [m]$. 
\end{defn}

Here is the key fact relating traces on $\bb F(n,m)$ and quantum commuting correlations:

\begin{fact}(\cite{vern})
For $p\in [0,1]^{n^2m^2}$, we have $p\in \cqc^s(n,m)$ if and only if there is a trace $\tau$ on $\bb F(n,m)$ that is adapted to $p$.
\end{fact}

\begin{proof}
In \cite{vern}, they show that $p\in \cqc^s(n,m)$ if and only if there is a tracial C*-algebra $(A,\tau)$ and a generating family of projections $p_{v,i}$ such that $\sum_{i=1}^m p_{v,i}=1$ for each $v=1,\ldots,n$ and such that $p(i,j|v,w)=\tau(p_{v,i}p_{w,j})$.  However, letting $\pi:C^*(\bb F(n,m))\to A$ be the surjective *-homomorphism determined by sending $e_{v,i}$ to $p_{v,i}$ and defining $\tau'$ on $C^*(\bb F(n,m))$ by $\tau'(x):=\tau(\pi(x))$, we obtain the equivalence with the above statement.
\end{proof}

% Given a trace $\tau$ on $\bb F(n,m)$ we let $p_\tau\in \cqc^s(n,m)$ denote the corresponding correlation matrix.

\begin{lem}
There is a computable function $s:\bb N^7\to \bigcup_{n,m\geq 2}\bb Q(i)\bb F(n,m)$ so that, for each $n,m\geq 2$, each $v,w\in [n]$, each $i,j\in [m]$, and each $k\geq 1$, we have that $s(v,w,i,j,k,n,m)\in \bb Q(i)\bb F(n,m)$ and $\|e^{n,m}_{v,i}e^{n,m}_{w,j}-s(v,w,i,j,k,n,m)\|<\frac{1}{k}$, where the norm is the norm on the universal representation of $\bb C \bb F(n,m)$, that is, the norm on $C^*(\bb F(n,m))$.
\end{lem}

\begin{proof}
% The statement of the theorem is that the elements $e^{n,m}_{v,i}$ are computable points of $C^*(\bb F(n,m))$ with respect to the presentation $(C^*(\bb F(n,m)),\bb F(n,m))$.  (See \cite[Section 2]{fox} for the terminology used in this proof.)  
The theorem follows from the fact that the relation $(u_v^{n,m})^j=\sum_{i=1}^m \xi_m^{ji} e^{n,m}_{v,i}$ mentioned above can be effectively inverted to express each $e^{n,m}_{v,i}$ as a polynomial in the generator $u_v^{n,m}$ with coefficients in the computable field $\bb Q(\xi_m)$.  Moreover, this procedure is uniform in $n$ and $m$.
% The lemma now follows from the fact that the norm on $\bb C(\bb F(n,m))$ is computable from above with respect to the above presentation, uniformly in $n$ and $m$; see \cite[Theorem 3.3]{fox}.  (This might also be implicit in the proof of \cite[Corollary 2.2]{canyou}.)
\end{proof}

In order to match notation, in what follows we will rewrite $s(v,w,i,j,k,n,m)$ as $s_{v,w,i,j,k}^{n,m}$.  Also, given a function $\tau:\bb F(n,m)\to \bb D$, we extend it to a function $\tau:\bb Q(i)\bb F(n,m)\to \bb C$ by linearity.

\begin{defn}
Fix $n,m\geq 2$, $p\in \bb [0,1]^{n^2m^2}$, and $k\geq 1$.  We say that a function (not necessarily a trace) $\tau:\bb F(n,m)\to \bb D$ is \textbf{$k$-adpated to $p$} if $$|p(i,j|v,w)-\tau(s^{n,m}_{v,w,i,j,k})|<\frac{1}{k}$$ for all $v,w\in [n]$ and all $i,j\in [m]$.  
\end{defn}

Here is the ``stability'' property satisfied by synchornous quantum commuting correlations:

\begin{prop}\label{lemma2}
Given $m,n\geq 2$ and $\epsilon>0$, there is $k\geq 1$ such that, for all $p\in [0,1]^{n^2m^2}$, if there is a $k$-approximate trace $\tau$ on $\bb F(n,m)$ that is $k$-adpated to $p$, then there is $p'\in \cqc^s(n,m)$ with $\|p-p'\|_\infty<\epsilon$. 
\end{prop}

\begin{proof}
Suppose that the lemma is false for some $m,n,\epsilon$, that is, for each $k\geq 1$, there is $p_k\in \bb [0,1]^{n^2m^2}$ for which there is a $k$-approximate trace $\tau_k$ on $\bb F(n,m)$ that is $k$-adapted to $p_k$ and yet $\|p_k-p\|_\infty\geq \epsilon$ for all $p\in \cqc^s(n,m)$.  Let $p$ be a subsequential limit of $p_k$.  Since $\tau_k$ is pointwise bounded, by passing to a subsequence if necessary, we may suppose that $\tau_k$ converges pointwise to a function $\tau:\bb F(n,m)\to \bb D$.  Note then that $\tau$ is an actual trace on $\bb F(n,m)$ that is adapted to $p$, whence $p\in \cqc^s(n,m)$.  Since $\|p-p_k\|_\infty<\epsilon$ for some $k
$, this is a contradiction.
\end{proof}

\section{The main theorem}

\begin{defn}
We say that $T:\bb N^2\to \bb N$ is a \textbf{qc-modulus} if it satisfies the conclusion of Proposition \ref{lemma2} for $\epsilon=\frac{1}{12}$, that is, for all $n,m\geq 2$ and all $p\in \bb [0,1]^{n^2m^2}$, if there is a $T(n,m)$-approximate trace $\tau$ on $\bb F(n,m)$ that is $T(n,m)$-adapted to $p$, then there is $p'\in C_{qc}^s(n,m)$ with $\|p-p'\|_\infty<\frac{1}{12}$.
\end{defn}

We let $X^{n,m}_k$ denote those $p\in \bb ([0,1]\cap \bb Q)^{n^2m^2}$ for which there is a $k$-approximate trace $\tau:\bb F(n,m)\to \bb Q(i)$ that is $k$-adpated to $p$.

The following lemma is clear:

\begin{lem}
Each $X^{n,m}_k$ is recursively enumerable, uniformly in $k$, $n$, and $m$.
\end{lem}

Here is the main result of this note:

\begin{thm}\label{maintheorem}
If there is a computable qc-modulus $T:\bb N^2\to \bb N$, then every language in $\mip^{co,s}$ is in RE (and thus decidable).
\end{thm}

\begin{proof}
Fix a computable qc-modulus $T$ and suppose that $L\in \mip^{co,s}$.  Let $z\mapsto \mathfrak G_z$ be an efficient mapping from strings to nonlocal games such that:
\begin{itemize}
    \item if $z\in L$, then $\pval^{co}(\mathfrak G_z)\geq \frac{2}{3}$
    \item if $z\notin L$, then $\pval^{co}(\mathfrak G_z)\leq \frac{1}{3}$.
\end{itemize}

Here is the algorithm for enumerating $L$.  Given a string $z$, first determine the dimensions $n$ and $m$ for $\frak G_z$.  Set $k:=T(n,m)$ and let $(p_l)$ be a computable enumeration of $X^{n,m}_k$.  The algorithm then simply starts computing $\val(\frak G_z,p_l)$ for each $l$; if for some $l$ we see that $\val(\frak G_z,p_l)> \frac{1}{2}$, then we declare that $z\in L$. 

Soundness of the algorithm:  suppose that $\val(\mathfrak G_z,p_l)>\frac{1}{2}$.  By the choice of $k$, there is $p\in \cqc^s(n,m)$ such that $\|p_l-p\|_\infty<\frac{1}{12}$.  It follows that $\val^{co,s}(\frak G_z)\geq \val(\frak G_z,p)> \frac{1}{2}-\frac{1}{12}>\frac{1}{3}$, which tells us that $z\in L$.

Completeness of the algorithm:  suppose that $z\in L$ and take $p\in \cqc^s(n,m)$ such that $\val(\frak G_z,p)\geq \frac{2}{3}$.  Let $\tau_p$ be a trace on $\bb F(n,m)$ that is adapted to $p$, whence $|p(i,j|v,w)-\tau_p(s^{n,m}_{v,w,i,j,k})|\leq \|e^{n,m}_{v,i}e^{n,m}_{w,j}-s^{n,m}_{v,w,i,j,k}\|<\frac{1}{k}$ for all $v,w\in [n]$ and all $i,j\in [m]$.  Fix $\eta>0$ small enough and let $p'\in ([0,1]\cap \bb Qm)^{n^2m^2}$ and $\tau':\bb F(n,m)\to \bb D\cap \bb Q(i)$ be such that $\|p-p'\|_\infty,\|\tau_p-\tau'\|_\infty<\eta$.  By Lemma \ref{lemma1}, if $\eta$ is small enough, then $\tau'$ is a $k$-approximate trace on $\bb F(n,m)$.  Note also that $$|p'(i,j|v,w)-\tau'(s^{n,m}_{v,w,i,j,k})|<2\eta+\|e^{n,m}_{v_i}e^{n,m}_{w,j}-s^{n,m}_{v,w,i,j,k}\|<\frac{1}{k}$$ as long as $\eta$ is small enough.  (Note that the conditions on $\eta$ are effective in terms of $k$ and thus in terms of $n$ and $m$.)  It follows that $\tau'$ is $k$-adapted to $p'$, whence there is $l$ such that $p'=p_l$.  As long as $\eta<\frac{1}{12}$, we have that $\val(\frak G_z,p_l)\geq \frac{2}{3}-\frac{1}{12}>\frac{1}{2}$.  Consequently, the algorithm will tell us that $z\in L$. 
% Take $\delta>0$ as in Lemma \ref{lemma1} for our fixed $k$ and a trace $\tau$ on $\bb F(n,m)$ that is adapted to $p$.  Let $\tau':\bb F(n,m)\to \bb D\cap \bb Q(i)$ be such that $\|\tau-\tau'\|_\infty<\delta$.  It follows that   Let $p'$ denote the correlation matrix corresponding to $\tau'$.  Note then that $\tau'$ is a $k$-approximate trace on $\bb F(n,m)$ that is $k$-adapted to $p$.  Take $l$ such that...  
% Without loss of generality, $\delta<\epsilon$.  Take $l$ such that $\|p-p_l\|_\infty<\delta$.  Then the value corresponding to $p_l$ is at least $\frac{2}{3}-\epsilon$, which is strictly larger than $r$.  Consequently, the algorithm will tell us that $z\in L$. 
\end{proof}

We note that it remains a possibility that $\mip^{co,s}$ is properly contained in co-RE and yet there is no computable qc-modulus.

% Note that one does not need to worry if the map $(m,n)\mapsto k(m,n):\bb N^2\to \bb N$ is computable.  We have one such algorithm for every choice of function $f:\bb N^2\to \bb N$ and we know that for one such function, namely for $f(m,n)=k(m,n)$, this algorithm works.  Said slightly differently, we can view the algorithm above as consisting of infinitely many algorithms, one for each pair $(m,n)$ of dimensions.  Thus, we are just effectively enumerating the games of each dimension for which this works and then we just check if the game $\frak G_z$ for $z$ is on the appropriate list.

\end{document}